\documentclass[11pt]{article}
\usepackage{amsthm}
\usepackage{amsmath}
\usepackage{amssymb}
\usepackage{geometry}
\usepackage{times}
\usepackage{paralist}
\usepackage[usenames]{color}
\usepackage{enumerate}
\usepackage{graphicx}
\usepackage{lscape}
\usepackage{rotating}
\usepackage{placeins}
\usepackage{authblk}

\renewcommand{\phi}{\varphi}
\newcommand{\set}[1]{\left\{#1\right\}}

\newcommand{\st}{\medspace | \medspace}
\newcommand{\hide}[1]{ }

\renewcommand{\phi}{\varphi}

\newcommand{\Top}{\mathsf{Top}}
\newcommand{\Bot}{\mathsf{Bot}}

\newcommand{\Mid}{\mathsf{Mid}}
\newcommand{\Clique}{\mathsf{Clique}}

\newcommand{\End}{\mathsf{End}}
\newcommand{\Tri}{\mathsf{Tri}}
\newcommand{\A}{\mathsf{A}}
\newcommand{\B}{\mathsf{B}}
\newcommand{\CONGEST}{\ensuremath{\mathsf{CONGEST}\ }}
\newcommand{\LOCAL}{\ensuremath{\mathsf{LOCAL}\ }}

\DeclareMathOperator*{\ex}{ex}

\theoremstyle{plain}
\newtheorem{theorem}{Theorem}[section]

\newtheorem{lemma}[theorem]{Lemma}

\newtheorem{property}{Property}
\newtheorem{corollary}[theorem]{Corollary}

\newtheorem{definition}{Definition}

\renewcommand{\include}{\input}

\begin{document}

\title{Superlinear Lower Bounds for Distributed Subgraph Detection}


\author{Orr Fischer}
\affil{Tel Aviv University\authorcr
  \texttt{orrfischer@mail.tau.ac.il}}
\author{Tzlil Gonen}
\affil{Tel Aviv University\authorcr
  \texttt{tzlilgon@mail.tau.ac.il}}
\author{Rotem Oshman}
\affil{Tel Aviv University\authorcr
  \texttt{roshman@tau.ac.il}}

\maketitle

\abstract{In the distributed subgraph-freeness problem, we are given a graph $H$, and asked to determine whether 
	the network graph contains $H$ as a subgraph or not. Subgraph-freeness is an extremely \emph{local} problem:
	if the network had no bandwidth constraints, we could detect any subgraph $H$ in $|H|$ rounds, by having
	each node of the network learn its entire $|H|$-neighborhood.
	However, when bandwidth is limited, the problem becomes harder.

	Upper and lower bounds in the presence of congestion have been established for several
	classes of subgraphs, including cycles, trees, and more complicated subgraphs.
	All bounds shown so far have been linear or sublinear.
	We show that the subgraph-freeness problem is not, in general, solvable in linear time:
	for any $k \geq 2$, there exists a subgraph $H_k$ such that $H_k$-freeness requires $\Omega( n^{2-1/k} / (Bk) )$
	rounds to solve. Here $B$ is the bandwidth of each communication link.
	The lower bound holds even for diameter-3 subgraphs and diameter-3 network graphs.
	In particular, taking $k = \Theta(\log n)$, we obtain a lower bound of $\Omega(n^2 / (B \log n))$.
}

\section{Introduction}
\label{sec:intro}

In the \emph{subgraph-freeness} problem, we are given a constant-size graph $H$, and the goal is to determine whether the network graph contains a copy of $H$ as a subgraph, or not.
Subgraph-freeness is an extremely \emph{local} problem: in the \LOCAL model of distributed computing, where network nodes can send messages of unrestricted size in each round, the subgraph-freeness for a graph $H$ of size $k$ can be solved in at most $O(k)$ rounds --- we simply have each node collect its entire $k$-neighborhood and check if it contains a copy of $H$.
However, in the \CONGEST model, where the bandwidth on each edge \emph{is} restricted, the picture is less clear.

Subgraph-freeness in the \CONGEST network model has recently received significant attention from the distributed computing community~\cite{ DFO14, triangle_free, ELM17, FGO17, FMORT17, square-free, KorhonenR17,GO17,Censor17}.
Several linear and sublinear upper and lower bounds have been shown for special cases, including triangles and larger cycles~\cite{DFO14,triangle_free,Censor17}; trees, which can be detected in $O(1)$ rounds~\cite{FGO17,FMORT17}; cliques and complete bipartite subgraphs, which can be detected in $O(n)$ rounds~\cite{DFO14}; and more complicated classes of graphs, for which in some cases we can prove a lower bound of the form $\Omega(n^{\delta})$, where $\delta < 1$~\cite{GO17}.
This raises the obvious question --- \emph{can any constant-sized subgraph $H$ be detected in $O(n)$ rounds?}

In this note we show that the answer is negative:
\begin{theorem}
For any $k, n \geq 2$, there is a graph $H_k$ of size $O(k)$ and diameter 3,
such that $H_k$-freeness requires $\Omega(n^{2-1/k}/(Bk))$ rounds in the \CONGEST model,
even when the network diameter is 3.
	\label{thm:hk_lower_bound}
\end{theorem}
Here, $B$ is the bandwidth parameter bounding the number of bits each node can send on each of its edges in one communication round.

Our lower bound provides another ``natural'' problem which has a superlinear lower bound in the \CONGEST model; to our knowledge, previously the only examples of such problems were found in \cite{CHKP17}.
Since the subgraph-freeness problem is also extremely local, we get a separation between the \CONGEST and \LOCAL models.
In particular, if we take $k = \Theta(\log n)$ in our lower bound, we obtain:
\begin{corollary}
	There is a graph $H$ of size $\Theta(\log n)$, such that $H$-freeness can be solved in $O(\log n)$ rounds in the \LOCAL model,
	but requires $\Omega(n^2/(B \log n))$ rounds in the \CONGEST model.
\end{corollary}
This is nearly the largest separation possible: in the \CONGEST model, if the network graph has diameter $D$, then $O(n^2/B+D)$ rounds suffice for each node to learn the entire network graph. We can then check whichever graph property we are interested in locally, and have each node output the answer.

Because our lower bound applies even when the network has diameter 3,
we can increase the separation slightly: we modify the problem so that nodes must check whether the network graph both
\begin{inparaenum}[(a)]
\item has diameter at most 3, and
\item contains a copy of $H$.
\end{inparaenum}
For $H$ of size $\Theta(\log n)$, the modified problem can be solved in 3 rounds in the \LOCAL model, but still requires $\Omega( n^2 / (B \log n))$ rounds in the \CONGEST model by Theorem~\ref{thm:hk_lower_bound}.
\begin{corollary}
	There is problem which can be solved in 3 rounds in the \LOCAL model,
	but requires $\Omega(n^2/(B \log n))$ rounds in the \CONGEST model.
\end{corollary}

\subsection{Overview of the Lower Bound}
We follow the classical approach of obtaining distributed lower bounds by reducing from \emph{two-party communication complexity} (see Chapter 8 of the textbook~\cite{kushilevitz_nisan}):
to show the lower bound on $H$-freeness, we show that if there were a fast algorithm for $H$-freeness in a certain class of networks,
then we could construct from it a communication-efficient two-party protocol for solving a ``hard'' function, in our case the set disjointness~\cite{KS92,Raz92} function.

The key aspect of the construction required for this approach is a \emph{sparse cut}: the two players simulate the distributed algorithm by partitioning the network into two parts, $A$ and $B$, with Alice locally simulating all the nodes in $A$, and Bob simulating the nodes in $B$. To carry out the simulation, the players must send each other the contents of messages sent across the cut between $A$ and $B$;
we require a sparse cut to get a communication-efficient protocol.

Our construction is inspired by the bit-gadget of \cite{ACK16,CHKP17}, which yields a graph with a very sparse cut - only $O(\log{n})$ edges.
Informally, the bit-gadget of~\cite{ACK16,CHKP17} corresponds to the \emph{binary encoding} of the numbers $\set{1,\ldots,n}$.
Here we use a different representation: we represent the numbers $\set{1,\ldots,n}$ as subsets of size $k$ of $\set{1,\ldots,kn^{1/k}}$,
and this gives us a cut of size $O(k n^{1/k})$. We embed a set disjointness instance of size $n^2$ in the graph, which requires $\Omega(n^2)$ bits of total communication; but since our cut is fairly sparse, simulating one round of the protocol requires only $O(k n^{1/k} \cdot B)$ bits,
so we get a lower bound of $\Omega(n^{2-1/k}/(Bk))$ on the number of rounds.

\subsection{Related Work}
The problem of subgraph-freeness (also called \emph{excluded} or \emph{forbidden subgraphs}) 
has been extensively studied in both the centralized and the distributed worlds.
For the general problem of detecting whether a graph $H$ is a subgraph of $G$, 
where both $H$ and $G$ are part of the input, the best known sequential algorithm is exponential \cite{Ullmann76}. 
When $H$ is fixed and only $G$ is the input, the problem becomes solvable in polynomial time. 

In the distributed setting, ~\cite{FGO17} and ~\cite{FMORT17} very recently provided constant-round randomized 
and deterministic algorithms, respectively, for detecting a fixed tree in the \CONGEST model.
Both papers, as well as several others ~\cite{Brakerski2011,triangle_free,square-free, ELM17} ,
also considered more general graphs, but with the exception of trees, they
studied the \emph{property testing} relaxation of the problem,
where we only need to distinguish a graph that is $H$-free from a graph that is \emph{far} from $H$-free.
(In~\cite{ELM17} there is also a property-testing algorithm for trees.)
Here we consider the \emph{exact}
version.

Another recent work~\cite{IzumiLG2017} gave randomized algorithms in the \CONGEST model for triangle 
detection and triangle listing, with round complexity $\widetilde{O}(n^{2/3})$ 
and $\widetilde{O}(n^{3/4})$, respectively, and also established a 
lower bound of $\widetilde{\Omega}(n^{1/3})$ on the round complexity of triangle listing.  
There is also work on testing triangle-freeness in the congested clique model~\cite{CHKKLPJ15,DLP12} and in other, less directly related distributed models.

As for lower bounds on $H$-freeness in the \CONGEST model, the only ones in the literature (to our knowledge) are for cycles~\cite{DFO14},
and the reductions from~\cite{GO17}, which construct hard graphs from other hard graphs by replacing their vertices or their edges 
with other graphs.
(In~\cite{DFO14} there are lower bounds for other graphs, in a broadcast variant of the \CONGEST model where nodes are required to send the \emph{same}
message on all their edges.)
For any fixed $k > 3 $, there is a polynomial lower bound for detecting the $k$-cycle $C_k$ in the \CONGEST model:
it was first presented by ~\cite{DFO14}, which showed that $\Omega(\ex(n,C_k)/B)$ rounds are required,
where $\ex(n,C_k)$ is the largest possible number of edges in a $C_k$-free graph over $n$ vertices (see~\cite{Furedi2013} for a survey on extremal graphs with forbidden subgraphs).
In particular, for odd-length cycles, the lower bound of ~\cite{DFO14} is nearly linear.
Very recently,~\cite{KorhonenR17} improved the lower bound for even-length cycles to $\Omega(\sqrt{n} / B)$.
All lower bounds mentioned above are linear or sublinear.

In a recent work \cite{CHKP17}, the first near-quadratic lower bounds (and indeed the first superlinear lower bounds) in CONGEST were shown. The problems addressed in~\cite{CHKP17} include some NP-hard problems such as minimum vertex cover and graph coloring, as well as a weighted variant of $C_8$-freeness, called \emph{weighted cycle detection}: given a weight $W \in [0,\mathrm{poly}(n)]$, the goal is to determine whether the graph has a cycle of length $8$ and weight exactly $W$. 

\section{Preliminaries}

\paragraph{Notation.}
We let $V(G), E(G)$ denote the vertex and edge set of graph $G$, respectively,
If $\sigma : V(H) \rightarrow V(G)$ is a mapping of vertices from $H$ into vertices of $G$,
then
\begin{itemize}
	\item For a subset $U \subseteq V(H)$, we let $\sigma(U) = \set{ \sigma(u) \st u \in U}$,
		and
	\item For a subgraph $H'$ of $H$, we let $\sigma(H')$ be the subgraph of $G$ induced by $\sigma( V(H') )$.
\end{itemize}

\paragraph{The \CONGEST model.}
The \CONGEST model is a synchronous network model, where computation proceeds in \emph{rounds}.
In each round, each node of the network may send $B$ bits on each of its edges, and these
messages are received by neighbors in the current round.
Typically, $B$ is taken to be polylogarithmic in the size $n$ of the network graph.

We are interested in the \emph{subgraph-freeness} problem, defined as follows:
\begin{definition}[Subgraph freeness]
	Fix a graph $H$.
In the \emph{$H$-freeness} problem,
the goal is to determine whether the input graph $G$ contains a copy of $H$ as a subgraph or not,
that is, whether there is a subset $U \subseteq V(G)$ such that the subgraph induced by $G$ on $U$ is isomorphic to $H$.
\end{definition}
We say that a distributed algorithm $A$ \emph{solves $H$-freeness with success probability $p$} if
\begin{itemize}
	\item When $A$ is executed in a graph containing a copy of $H$, the probability that all nodes accept is at least $p$.
	\item When $A$ is executed in an $H$-free graph, the probability that at least one node rejects is at least $p$.
\end{itemize}
We typically assume constant $p$, e.g., $p = 2/3$.

\paragraph{Two-party communication complexity.}
Our lower bound is shown by reduction from \emph{two-party communication complexity}:
we have two players, Alice and Bob, with private inputs $X,Y$, respectively.
The players wish to compute a joint function $f(X,Y)$ of their inputs, and we are interested in the total number of bits they must exchange to do so (see the textbook~\cite{kushilevitz_nisan} for more background on communication complexity).

In particular, we are interested in the \emph{set disjointness} function, where the inputs $X,Y$ are interpreted as subsets $X,Y \subseteq [n]$, and the goal of the players is to determine whether $X \cap Y = \emptyset$.
The celebrated lower bound of~\cite{KS92,Raz92} shows that even for randomized communication protocols, the players must exchange $\Omega(n)$ bits to solve set disjointness with constant success probability.

\paragraph{Vertex names.}
Our graph constructions are somewhat complicated. To avoid a multitude of subscripts and superscripts,
we assign vertices names that are tuples, indicating which ``logical part'' of the graph they belong to 
and their index inside that part.
We also adopt the following convention with respect to vertex names:
if $X = \set{ (X_1, \ldots), \ldots, (X_{\ell},\ldots) }$ is a set of vertices, then $X'$ denotes the same vertices but using $X_i'$ instead of $X_i$ in the name:
$X' = \set{ (X_1', \ldots), \ldots, (X_{\ell}',\ldots) }$ (only the first coordinate of each vertex name is modified, the rest remain the same).
Finally, we ``flatten'' the parentheses nesting of tuples, so that instead of, e.g., $(1, (2,3))$ we always write $(1,2,3)$.

\section{The Lower Bound}

Our lower bound shows that for any $k \geq 2$, there is a graph on $O(k)$ vertices which requires $\Omega(n^{2-1/k}/(Bk))$ rounds to detect.
We do not require $k$ to be constant; it may grow with $n$.
Furthermore, $H_k$ has diameter 3, and the lower bound applies even when the network in which we wish to detect $H_k$ also have diameter 3.

We begin by constructing the graph $H_k$.

\subsection{Informal Description of $H_k$}
Informally, $H_k$ has several ``parts'' (see Fig.~\ref{fig:Hk} for an illustration of $H_k$):
\paragraph{Cliques.}
	We put in five cliques, one of each size $s = 6,\ldots,10$.
		We pick one special vertex $v_s$ from each clique, and connect $\set{ s_6,\ldots,s_{10}}$ in a 5-clique.
		Each remaining vertex of $H_k$ is connected to exactly one special clique vertex $v_s$,
		and no other clique vertices.

		This serves two purposes:
		\begin{enumerate}
			\item It reduces the diameter of $H_k$ to 3: each vertex in $H_k$ is connected to some special clique vertex $v_s$,
				and all special clique vertices are connected to each other.
			\item Each ``part'' of $H_k$ (except for the cliques themselves) is connected to exactly 
				one $s$-clique for $s \in \set{6,\ldots,10}$,
				so the cliques serve to ``mark'' the different parts of $H_k$.

				When we construct the network graph $G$ in which we show the lower bound on $H_k$-freeness,
				we will make sure $G$ also contains exactly one copy of each $s$-clique for $s = 6,\ldots,10$,
				so that any isomorphism mapping $H_k$ into $G$ must map the $s$-clique of $H_k$ onto the $s$-clique in $G$.
				The ``parts'' of $G$ will echo the ``parts'' of $H_k$, and the connections between the cliques
				and the other vertices will force any isomorphism from $H_k$ into $G$ to respect 
				this logical partition.
		\end{enumerate}
\paragraph{Top and bottom.}
	The remainder of $H_k$ consists of two identical copies of a graph $H$.
		We call the two copies ``top'' and ``bottom'', respectively.
	The subgraph $H$ consists of $k$ triangles $\Tri_1,\ldots,\Tri_k$, and two additional ``endpoint nodes'', which we call $\A$ and $\B$.
	In each triangle $\Tri_i$, we have three vertices denoted $(i, \A), (i,\B), (i, \Mid)$.
	Endpoint $\A$ is connected to all triangle vertices of the form $\set{ (i, \A) \st i \in [k]}$,
	endpoint $\B$ is connected to all triangle vertices of the form $\set{ (i,\B) \st i \in [k]}$,
	and the ``middle vertices'' of the triangles are not connected to either endpoint. Finally, the top and bottom A endpoints are connected by an edge, and the top and bottom B endpoints are connected by an edge.

\begin{sidewaysfigure}[h!]
 \centering
  \includegraphics[width=0.9\textwidth]{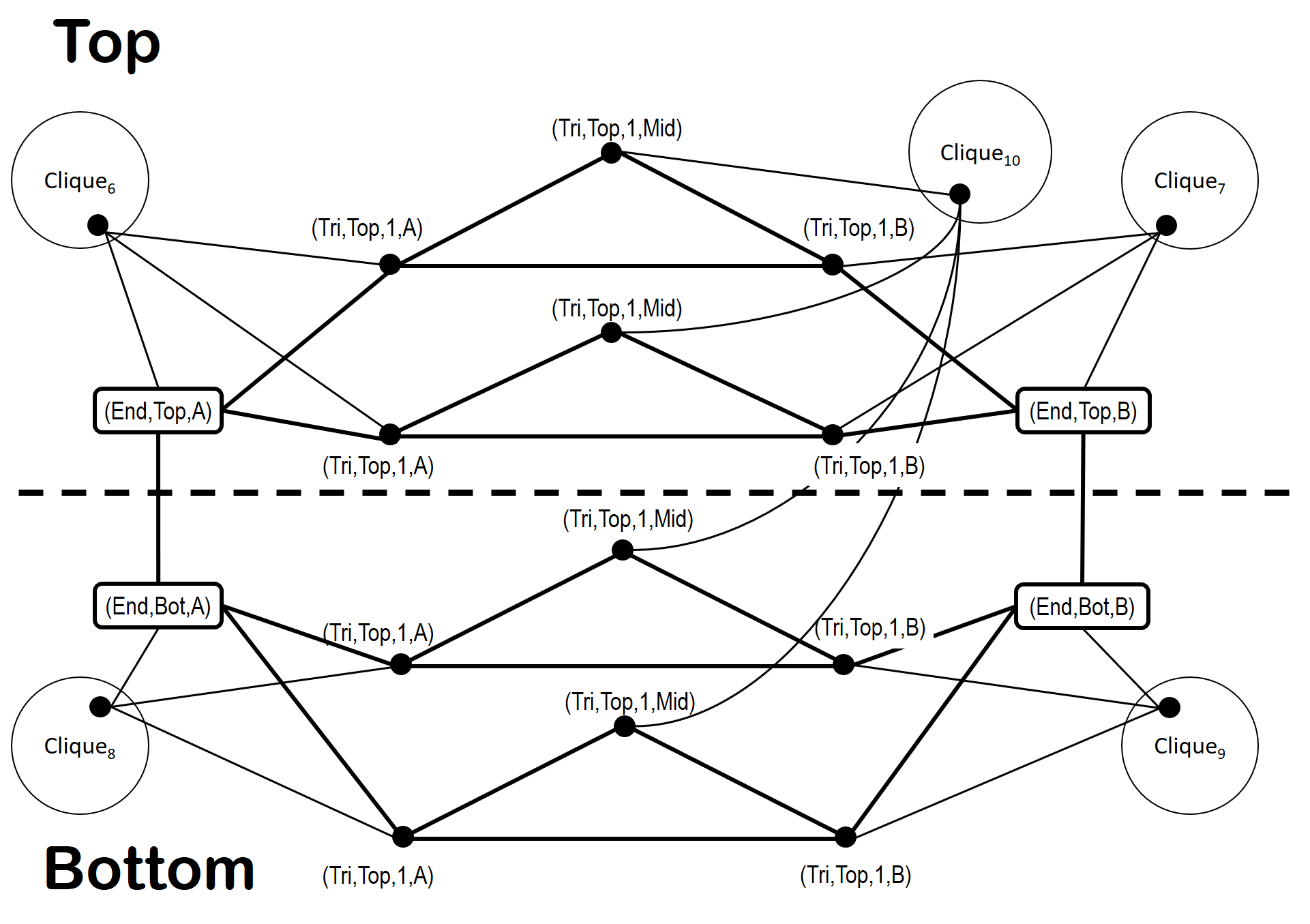}
  \caption{The graph $H_k$ for $k = 2$. 
	  In addition to the edges shown here, the five clique nodes $\set{(\Clique_s, 0) \st 6 \leq s \leq 10}$ are connected to each other; we omit these edges from the figure for clarity.}
  \label{fig:Hk}
  \end{sidewaysfigure}
	
\subsubsection{Our naming scheme}
We name each vertex of $H_k$ by a tuple indicating its role (a clique node, an endpoint, or a triangle node),
whether it belongs to the top or the bottom part, its ``orientation'' ($\A / \B / \Mid$),
and its index.

\FloatBarrier 

\subsection{Formal Definition of $H_k$}
	For any integer $k \geq 1$, the graph $H_k$ contains the following components%
	\footnote{We use ``component'' simply to mean a logical unit, not a maximal connected component of the graph.}
	\begin{itemize}
		\item Two edges, denoted $\set{ (\End, \Top, \A), (\End, \Bot, \A)}$ and $\set{( \End, \Top, \B), (\End, \Bot, \B)}$ respectively,
		\item For each $s \in \set{ 6, 7, 8, 9, 10}$, an $s$-clique denoted $\Clique_s = \set{\Clique_s} \times [s]$,
		\item $2k$ triangles indexed by $\set{\Top, \Bot} \times [k]$, with triangle $(S, i) \in \set{\Top, \Bot} \times [k]$ denoted
			\begin{equation*}
				\Tri_{S,i} = \set{ (\Tri, S, i, X) \st X \in \set{ \A, \B, \Mid}}.
			\end{equation*}
			For $X \in \set{\A, \B, \Mid}$, we let  $\Tri_{S,X} = \set{ (\Tri, S, i, X) \st i \in [k]}$.
	\end{itemize}
	For convenience, we associate the various cliques with ``directions'' (top/bottom $\times$ $\A$/$\B$/$\Mid$) as follows: we denote
	\begin{equation*}
		c_{S, X} = 
		\begin{cases}
			6, \quad \text{if $S = \Top, X = \A$},\\
			7, \quad \text{if $S = \Top, X = \B$},\\
			8, \quad \text{if $S = \Bot, X = \A$},\\
			9, \quad \text{if $S = \Bot, X = \B$},\\
			10, \quad \text{if $X = \Mid$}.
		\end{cases}
	\end{equation*}
	The components are connected to each other by the following additional edges:
	\begin{enumerate}[(E1)]
		\item For each ``direction'' $(S,P) \in \set{ \Top, \Bot } \times \set{ \A, \B}$, the endpoint $(\End, S, P)$ is connected to vertex $(\Clique_{c_{S,P}}, 0)$ of the $c_{S,P}$-clique.
		\item For each ``direction'' $(S,P) \in \set{ \Top, \Bot } \times \set{ \A, \B, X}$, all triangle vertices in $\Tri_{S,X}$ are connected to vertex $(\Clique_{c_{S,P}}, 0)$ of the $c_{S,P}$-clique.
		\item The 0-vertices of the five cliques are connected to each other, that is, we connect all vertices in $\set{ (\Clique_s, 0) \st s \in \set{ 6,\ldots,10}}$ to each other.
		\item For each $S \in \set{ \Top, \Bot}, P \in \set{\A,\B}$, the endpoint $(\End, S, P)$ is connected to all $k$ triangle nodes $(\Tri, S, i, P)$ for $i \in [k]$.
	\end{enumerate}

The graph $H_k$ has the following properties:
	\begin{property}
		The size of $H_k$ is $30+6k = O(k)$.
	\end{property}
	\begin{property}
		For each $k \geq 2$, the diameter of $H_k$ is 3.
	\end{property}
	\begin{proof}
		Each vertex of $H_k$ is connected by an edge to a vertex $(\Clique_s, 0)$ for some $s \in \set{ 6,\ldots,10}$, and in turn, the five vertices $(\Clique_s, 0)$ for $s \in \set{6,\ldots,10}$ are connected to each other.
		(The diameter is not smaller than 3, because, e.g., of the triangles.)
	\end{proof}

Next we describe the family of graphs in which we show that $H_k$-detection is hard.

\subsection{The Lower Bound Family $\mathcal{G}_{k,n}$}

Let $m = k\lceil n^{1/k} \rceil$.
Fix an ordering $Q_1,\ldots, Q_N$ of the subsets of size $k$ of $[ m ]$,
where $N = { m \choose k } = { k\lceil n^{1/k} \rceil \choose k}$.
Note that
\begin{equation*}
	N = { k \lceil n^{1/k} \rceil \choose k } \geq \left( \frac{ k n^{1/k} } {k}  \right)^k = \frac{ k^k n}{k^k} = n.
\end{equation*}
For each $i \in [N]$, let us denote 
$Q_i = \set{q_i^1,\ldots,q_i^k}$.
(We will only use the first $n$ subsets, $Q_1,\ldots,Q_n$.)

Our lower bound graph family $\mathcal{G}_{k,n}$ is defined as follows (see Fig.~\ref{fig:G}):
\begin{definition}[The graph family $\mathcal{G}_{k,n}$]
	Fix integers $k, n$.
	A graph $G$ is in the family $\mathcal{G}_{k,n}$ 
	if it has the following structure:
	\begin{itemize}
		\item The graph contains the following components:
			\begin{itemize}
				\item $n$ ``potential top endpoints'' and ``potential bottom endpoints'' of $H_k$:
					\begin{equation*}
						\End' \times [n].
					\end{equation*}
					For $S \in \set{\Top, \Bot}$ and $P \in \set{ \A, \B}$, we denote
					\begin{equation*}
						\End'_{S,P} = \set{ (\End', S, P, i) \st i \in [n]}.
					\end{equation*}
				\item $2m$ triangles, indexed by $\set{ \Top, \Bot} \times [m]$: for each $(S, i) \in \set{\Top,\Bot} \times [m]$,
					\begin{equation*}
						\Tri'_{S,i} = \set{ (\Tri', S, i, X ) \st X \in \set{\A,\B,\Mid}, i \in [m]}.
					\end{equation*}
					For $S \in \set{\Top, \Bot}$ and $X \in \set{ \A, \B, \Mid}$, we denote
					\begin{equation*}
						\Tri'_{S, X} = \set{ (\Tri', S, i, X) \st i \in [m]}.
					\end{equation*}
				\item Copies of each of the cliques in $H_k$: for each $s \in \set{6,7,8,9,10}$,
					\begin{equation*}
						\Clique_s' = \set{ (\Clique'_s, i) \st i \in [s] }.
					\end{equation*}
			\end{itemize}
		\item The graph contains the following edges between the components, and no other edges:
			\begin{itemize}
				\item For each ``direction'' $(S,P) \in \set{ \Top,\Bot} \times \set{\A,\B}$, all vertices in $\End'_{S,P}$ are connected to vertex $(\Clique'_{c_{S,P}}, 0)$ of the $c_{S,P}$-clique.
				\item For each ``direction'' $(S,X) \in \set{ \Top,\Bot} \times \set{\A,\B,\Mid}$, all vertices in $\Tri'_{S,X}$ are connected to vertex $(\Clique'_{c_{S,P}}, 0)$ of the $c_{S,P}$-clique.
				\item The 0-vertices of the five cliques are connected to each other, that is, we connect all vertices in $\set{ (\Clique_s', 0) \st s \in \set{ 6,\ldots,10}}$ to each other.
				\item For each $S \in \set{ \Top, \Bot}, P \in \set{\A,\B}$,
					and $i \in [n]$,
					the endpoint-copy $(\End', S, P, i)$ is connected to each of the $k$ triangle nodes $(\Tri, S, j, P)$ where $j \in Q_i$.
				\item Finally, for each $P \in \set{\A,\B}$, the graph may contain an arbitrary subset of the edges
					in $\End'_{\Top, P} \times \End'_{\Bot, P}$.
			\end{itemize}
	\end{itemize}
\end{definition}

\begin{sidewaysfigure}[h!]
 \centering
  \includegraphics[width=0.9\textwidth]{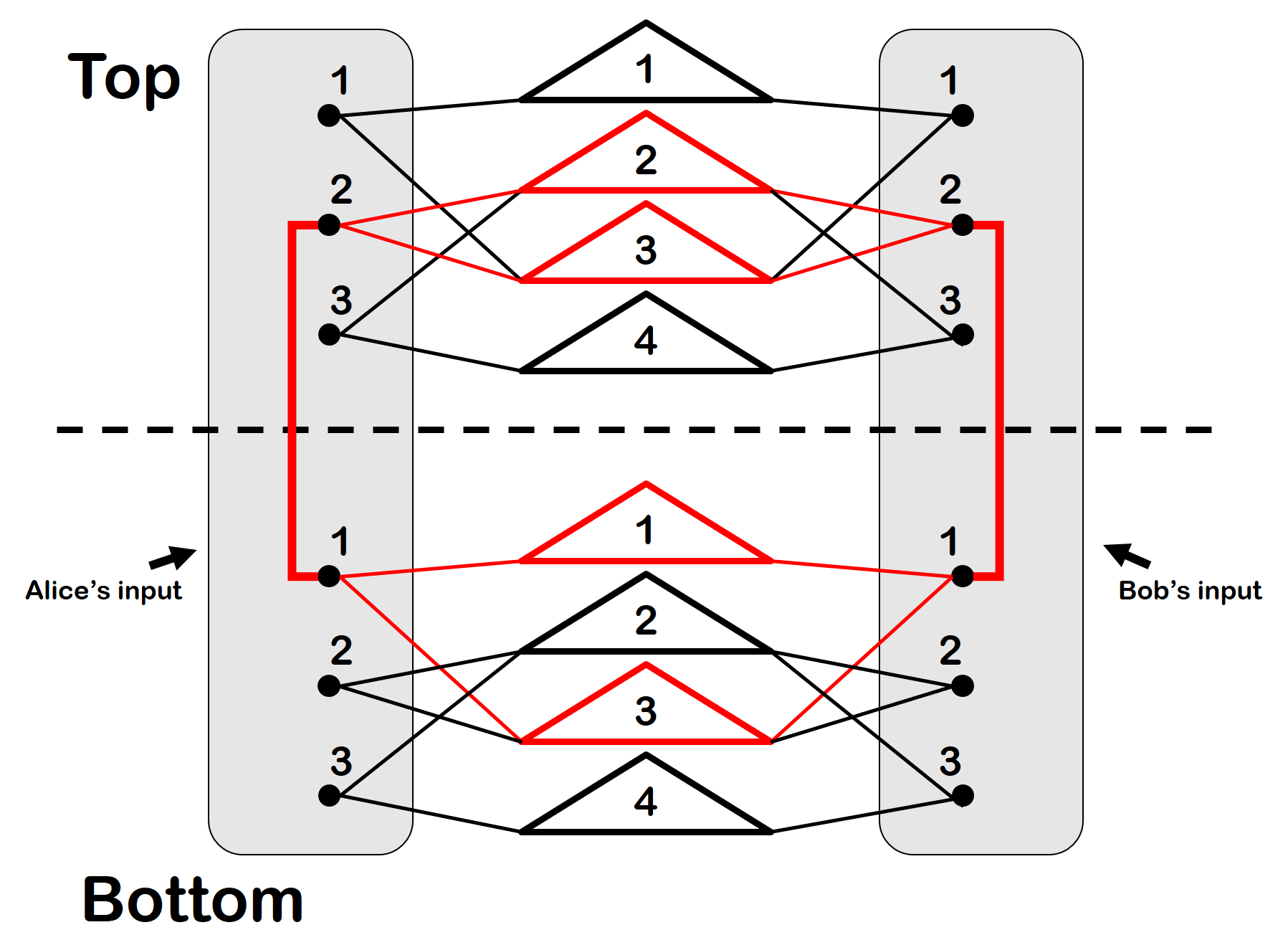}
  \caption{The graph $G_{X,Y} \in \mathcal{G}_{k,n}$ for $n = 3, k = 2$, so that $m = k \lceil n^{1/k} \rceil = 2 \cdot \lceil 3^{1/2} \rceil = 4$.
  For clarity, the figure omits the five cliques $\set{ \Clique'_s \st 6 \leq s \leq 10}$ and their edges.
  The inputs $X,Y$ whose graph $G_{X,Y}$ is depicted here both include $(2,1)$ (that is $(2,1) \in X \cap Y$).
  As a result, a copy of $H_k$ appears in $G_{X,Y}$, highlighted in red.
  }
  \label{fig:G}
  \end{sidewaysfigure}
The construction has the following properties.

\begin{property}
	Any graph in $\mathcal{G}_{k,n}$ has diameter 3.
\end{property}
\begin{proof}
	As with $H_k$, each vertex is connected to at least one of the clique vertices, $(\Clique'_s, 0$ for some $s \in \set{6,7,8,9,10}$,
	and these vertices are in turn all connected to each other.
\end{proof}

\begin{property}
	The size of each graph in $\mathcal{G}_{k,n}$
	is $4n + 6m + 30 \leq 4n + 6k(n^{1/k}+1) + 30 = O(n)$.
\end{property}

\begin{property}
	Fix $G \in \mathcal{G}_{k,n}$, and let $U \subseteq V(G)$ be a set of $s$ nodes, $s \in \set{6,7,8,9,10}$,
	such that the subgraph induced on $U$ by $G$ is the $s$-clique.
	Then $U \subseteq \Clique'_{s'}$ where $s' \geq s$.
	\label{prop:clique}
\end{property}

\begin{property}
	Any triangle in $G \in \mathcal{G}_{k,n}$ either includes some vertex of $\bigcup_{s =6}^{10} \Clique'_s$,
	or it is one of the triangles $\Tri'_{S,i}$ for $S \in \set{\Top,\Bot}, i \in [m]$.
	\label{prop:triangle}
\end{property}

\begin{property}
	Fix $G \in \mathcal{G}_{k,n}$ and $(S,P) \in \set{ \Top,\Bot } \times \set{ \A,\B}$.
	If $v \in V(G) \setminus \bigcup_{s = 6}^{10} \Clique'_s$ is a non-clique vertex adjacent to at least two
	distinct vertices in $\Tri'_{S,P}$,
	then $v \in \End'_{S,P}$.
	\label{prop:two_tri}
\end{property}
\begin{proof}
	Since we have ruled out the clique vertices,
	$\bigcup_{s = 6}^{10} \Clique'_s$,
	the remaining possibilities are the endpoint-copies, $\bigcup_{S,P} \End'_{S,P}$, and the triangle vertices.
	However, there are no edges in $G$ between distinct triangles, and inside each triangle we only have one vertex from $\Tri'_{S,P}$,
	while $v$ has edges to two vertices from $\Tri'_{S,P}$.
	Therefore $v$ cannot be a triangle vertex, and must be an endpoint-copy, $v = (\End', S', P', i) \in \End'_{S',P'}$.
	But $v$ would only be connected to vertices from $\Tri'_{S,P}$ in $G$ if $S' = S$ and $P' = S$,
	so in fact $v \in \End'_{S,P}$, as we claimed.
\end{proof}

\begin{lemma}
	A graph $G \in \mathcal{G}_{k,n}$ contains $H_k$ as a subgraph iff there exist $i_{\Top},i_{\Bot} \in [n]$
	such that
	\begin{equation*}
		( \End', \Top, \A, i_{\Top}), ( \End', \Bot, \A, i_{\Bot}) \in E(G)
	\end{equation*}
	and
	\begin{equation*}
		( \End', \Top, \B, i_{\Top}), (\End', \Bot, \B, i_{\Bot}) \in E(G).
	\end{equation*}
	\label{lemma:subgraph}
\end{lemma}
\begin{proof}
	First, suppose there exist such $i_{\Top},i_{\Bot} \in [n]$.
	Then $G$ contains a copy of $H_k$, witnessed by the isomorphism $\sigma$ that maps
	\begin{itemize}
		\item Each endpoint $(\End, S, P)$ of $H_k$ to the endpoint-copy $(\End', S, P, i_S)$ in $G$.
		\item Each clique $\Clique_s$ of $H_k$ to the copy $\Clique'_s$ in $G$, with $(\Clique_s,i)$ mapped to $(\Clique', i)$ for each $i \in [s]$.
		\item Each triangle $\Tri_{S, j}$ of $H_k$ to the triangle $\Tri'_{S, q_{i_S}^j}$ of $G$,
			with each vertex $(\Tri, S, j, X)$ mapped to $(\Tri', S, q_{i_S}^j, X)$ for each $X \in \set{\A,\B,\Mid}$.
	\end{itemize}
	Each triangle $\Tri_{s,j}$ of $H_k$ is mapped onto a triangle of $G$, and each clique $\Clique_s$ of $H_k$ is mapped onto a clique of $G$.
	The two endpoint-edges of $H_k$ are also mapped onto edges of $G$, because for each $P \in \set{ \A, \B}$,
	\begin{equation*}
		\sigma( \set{ (\End, \Top, P), (\End, \Bot, P) } ) = \set{ (\End', \Top, P, i_{\Top}), (\End', \Bot, P, i_{\Bot})},
	\end{equation*}
	which is an edge of $G$ by assumption.
	Let us go over the additional edges of $H_k$ and show that $\sigma$ maps them onto edges of $G$:
	\begin{enumerate}[(E1)]
		\item Edges of the form $e = \set{ (\End, S, P), (\Clique_{c_{S,P}}, 0)}$: 
			then
			\begin{equation*}
				\sigma(e) \in \End'_{S,P} \times \set{ \Clique'_{c_{S,P}}, 0 } \subseteq E(G).
			\end{equation*}
		\item Edges of the form $e = \set{ (\Tri, S, i, X), (\Clique_{c_{S,X}}, 0)}$:
			then
			\begin{equation*}
				\sigma(e) \in \Tri'_{S,X} \times \set{ \Clique'_{c_{S,P}}, 0 } \subseteq E(G).
			\end{equation*}
		\item Edges of the form $e = \set{ (\Clique_s, 0), (\Clique_t, 0)}$ for $s, t \in \set{6,\ldots,10}$:
			then
			\begin{equation*}
				\sigma(e) = \set{ (\Clique_s', 0), \Clique_t', 0)} \in E(G).
			\end{equation*}
		\item Edges of the form $e = \set{ (\End, S, P), (\Tri, S, j, P) }$ where $j \in [k]$:
			then
			\begin{equation*}
				\sigma(e) = \set{ (\End', S, P, i_S), (\Tri', S, q_{i_S}^j, P)} \in E(G).
			\end{equation*}
	\end{enumerate}

	For the other direction, suppose that $G$ does contain $H_k$ as a subgraph, and let $\sigma$ be an isomorphism mapping $H_k$
	onto its copy in $G$.

	By Property~\ref{prop:clique}, $\sigma$ must map each $s$-clique $\Clique_s$ in $H_k$ into some $s'$-clique $\Clique'_{s'}$, with $s' \geq s$.
	In particular, we must have $\sigma(\Clique_{10}) = \Clique'_{10}$, and therefore also $\sigma(\Clique_9) = \Clique'_9$,
	and so on, so that for each $s \in \set{6,\ldots,10}$ we have $\sigma(\Clique_s) = \Clique'_s$.

	Next, by Property~\ref{prop:triangle}, since we have already ``used up'' all the vertices in $\bigcup_{s = 6}^{10} \Clique'_s$,
	we know that
	$\sigma$ must map each triangle $\Tri_{S, i}$ of $H_k$
	onto some triangle $\Tri'_{S', i'}$ of $G$.
	Furthermore, if $\sigma( \Tri, S, i, X ) = (\Tri', S', i', X')$, then
	\begin{enumerate}
		\item $S = S'$: in $H_k$, vertex $(\Tri, S, i, \A) \in \Tri_{S,i}$ is connected to $\Clique_{c_{S,\A}}$ and none of the other cliques,
			and in $G$ the same holds for vertex $(\Tri, S', i', \A)$ and $\Clique_{c_{S', \A}}$.
			Since we know that $\sigma(\Clique_{c_{S,\A}}) = \Clique'_{c_{S,\A}}$,
			we must have $S = S'$, that is, $\Tri_{S,i}$ must be mapped onto $\Tri'_{S,i'}$ for some $i'$.
		\item $X = X'$: for a similar reason --- in $H_k$, vertex $(\Tri, S, i, X)$ is connected to $\Clique_{c_{S,X}}$ and no other clique,
			and the same in $G$ with $(\Tri, S', i', X')$ and $\Clique'_{c_{S',X'}}$;
			if $X' \neq X$ then $c_{S,X} \neq c_{S',X'}$,
			and since $\sigma(\Clique_{c_{S,X}}) = \Clique'_{c_{S,X}}$,
			we must have $X = X'$.
	\end{enumerate}
	So, for each triangle $\Tri_{S,i}$ of $H_k$,
	there is some index in $[m]$, which we abuse notation by denoting $\sigma(S,i)$,
	such that $\sigma( \Tri, S, i, X ) = (\Tri', S, \sigma(S,i), X)$.
	Let us denote, for $S \in \set{ \Top, \Bot}$,
	\begin{equation*}
		T_S = \set{ \sigma(S,i) \st i \in [k] }.
	\end{equation*}
	That is, $T_S$ is a collection of $k$ indices, $T_S \subseteq [m]$, indicating which triangles of $G$
	the $S$-triangles of $H_k$ were mapped to.

	Now consider the four endpoints of $H_k$.
	Each endpoint $(\End, S, P) \in V(H_k)$ is connected in $H_k$ to $k \geq 2$ triangle nodes, $\Tri_{S,P}$.
	We have seen that $\sigma(\Tri_{S,P}) \subseteq \Tri'_{S,P}$.
	Therefore, in $G$, $\sigma(\End, S, P)$ is connected to $k \geq 2$ distinct nodes from $\Tri'_{S,P}$;
	by Property~\ref{prop:two_tri},
	we must have $\sigma(\End, S, P) = (\End', S, P, i)$ for some $i \in [n]$.
	Let us further abuse notation by denoting $\sigma(\End,S,P) = i$ in this case.

	The only nodes in $\Tri'_{S,P}$ to which the endpoint-copy $(\End', S, P, i)$ is connected
	are the $k$ triangle vertices $(\Tri', S, j, X)$ where $j \in Q_i$.
	Therefore, $\sigma$ must map the $k$ triangles $\set{ \Tri_{S,i} \st i \in [k]}$ onto the $k$ triangles
	$\set{ \Tri'_{S,j} \st j \in Q_i }$;
	in other words, for each $S \in \set{\Top, \Bot}$, we must have $T_S = Q_{\sigma(\End, S, P)}$,
	and this holds for both $P = \A$ and $P = \B$,
	so in particular, $Q_{\sigma(\End, S, \A)} = Q_{\sigma(\End, S, \B)}$.
	But this implies that for both $S \in \set{\Top, \Bot}$ we have $\sigma(\End, S, \A) = \sigma(\End, S, \B)$,
	because for any $j_1 \neq j_2$ we have $Q_{j_1} \neq Q_{j_2}$.
	Let $i_{\Top} = \sigma(\End, \Top, \A) = \sigma(\End, \Top, \B)$
	and let
	$i_{\Bot} = \sigma(\End, \Bot, \A) = \sigma(\End, \Bot, \B)$.

	To conclude the proof, it remains to observe that since in $H_k$ the top-$\A$ and bottom-$\A$ endpoints are connected to each other
	(that is, $\set{ (\End, \Top, \A), (\End, \Bot, \A)} \in E(H_k)$),
	$\sigma$ must map these endpoints onto endpoint-copies that have an edge in $G$:
	\begin{equation*}
		\set{ (\End', \Top, \A, i_{\Top}), (\End', \Bot, \A, i_{\Bot})} \in E(G).
	\end{equation*}
	Similarly,
	since in $H_k$ the top-$\B$ and bottom-$\B$ endpoints are connected to each other
	we must have 
	\begin{equation*}
		\set{ (\End', \Top, \B, i_{\Top}), (\End', \Bot, \B, i_{\Bot})} \in E(G).
	\end{equation*}
	Therefore there exist indices $i_{\Top}, i_{\Bot}$ as required.

\end{proof}

\FloatBarrier

\begin{proof}[Proof of Theorem~\ref{thm:hk_lower_bound}]
	The proof is by reduction from set disjointness on the universe $[n]^2$.

	Fix an algorithm $A$ for solving $H_k$-freeness in the class $\mathcal{G}_{k,n}$,
	and let us construct from $A$ a protocol for disjointness.
	Given inputs $X,Y \subseteq [n]^2$, Alice and Bob construct a graph $G_{X,Y} \in \mathcal{G}_{k,n}$.
	The only freedom when constructing a graph in $\mathcal{G}_{k,n}$ is the choice of the edges we take
	from $\End'_{\Top,\A} \times \End'_{\Bot, \A}$ and from $\End'_{\Top,\B} \times \End'_{\Bot, \B}$.
	For this choice the players use their inputs:
	\begin{itemize}
		\item Edge $\set{(\End',\Top, \A, i), (\End', \Bot, \A, j)}$ is included in $G_{X,Y}$ iff $(i,j) \in X$, and
		\item Edge $\set{(\End',\Top, \B, i), (\End', \Bot, \B, j)}$ is included in $G_{X,Y}$ iff $(i,j) \in Y$.
	\end{itemize}
	By Lemma~\ref{lemma:subgraph}, the graph $G_{X,Y}$ includes a copy of $H_k$ as a subgraph iff
	$X \cap Y \neq \emptyset$.
	Thus, to solve their disjointness instance, Alice and Bob simulate the execution of $A$ in $G_{X,Y}$,
	and output ``$X \cap Y = \emptyset$'' iff $A$ rejects.

	Let us describe the simulation.
	We partition $V(G_{X,Y})$ into three parts:
	\begin{itemize}
		\item Alice's part, $V_A = \bigcup_{S \in \set{ \Top, \Bot}} \left( \End'_{S, \A} \times [n] \cup \Tri'_{S,\A} \right) \cup \Clique'_6 \cup \Clique'_8$,
		\item Bob's part, $V_B = \bigcup_{S \in \set{ \Top, \Bot}} \left( \End'_{S, \B} \times [n] \cup \Tri'_{S,\B} \right) \cup \Clique'_7 \cup \Clique'_9$,
		\item The shared part, $U$, comprising the remaining vertices, $U = \Tri'_{\Top, \Mid} \cup \Tri'_{\Bot, \Mid} \cup \Clique'_{10}$.
	\end{itemize}
	Note that each player knows all edges of $G_{X,Y}$, except those edges that are internal to the other player's part
	--- those are the only edges that depend on the other player's input.
	For example, Alice knows all edges in $\left( V_A \cup V_B \cup U \right) \times \left( V_A \cup U \right)$.

	The cut between $V_A$ and the rest of the graph, $V_B \cup U$, contains exactly two edges from each triangle $\Tri'_{S,i}$, and the six clique edges going from $(\Clique'_6, 0)$ or from $(\Clique'_8, 0)$ to another clique vertex $(\Clique'_s, 0)$ for $s \in \set{7,9,10}$.
	The size of this cut is $2 \cdot 2m + 6 = O(k n^{1/k})$.
	The cut between $V_B$ and $V_A \cup U$ is of the same size.

	To carry out the simulation, each player locally simulates the states of the nodes in its part of the graph, and also in the shared part.
	As we said above, each player knows all edges between the nodes it needs to simulate,
	so it knows which node is supposed to receive messages from which other nodes.

	To simulate one round of the algorithm, each player computes all the messages sent by the nodes it simulates, using shared randomness to agree on the random choices of the nodes in the shared part.
	Alice and Bob then send each other the messages sent on edges crossing the cut from their part of the graph to the rest of the graph.
	This costs $O( kn^{1/k} \cdot B)$ bits in total.
	The players then feed in the messages received by every node they simulate locally, and compute the next state of the node, again using
	shared randomness for the shared part.

	If $A$ runs for $R$ rounds, then the total cost of the simulation is $O(R \cdot kn^{1/k} \cdot B)$.
	Solving disjointness on $[n]^2$ requires $\Omega(n^2)$ bits.
	Thus, we must have
	\begin{equation*}
		R = \Omega\left( n^2 / \left( k n^{1/k} \cdot B  \right) \right) = \Omega( n^{2-1/k} / (Bk)).
	\end{equation*}

\end{proof}

\bibliographystyle{plain}
\bibliography{bibliography}
\end{document}